\documentclass[reqno,12pt]{amsart}

\usepackage{amsmath}
\usepackage{latexsym}
\usepackage{amssymb}
\usepackage{hyperref}
\usepackage{graphicx}
\usepackage{cite}
\usepackage{epstopdf}
\usepackage{ifpdf}
\ifpdf
\fi

\makeatletter
 
 \newcommand{\Rmnum}[1]{\expandafter\@slowromancap\romannumeral #1@}
 \makeatother

\newtheorem{theorem}{Theorem}[section]

\newtheorem{lemma}[theorem]{Lemma}

\newcommand{\C}{{\mathbb C}}

\newcommand{\be}{\begin{equation}}
\newcommand{\ee}{\end{equation}}
\newcommand{\bea}{\begin{eqnarray}}
\newcommand{\eea}{\end{eqnarray}}
\newcommand{\ba}{\begin{array}}
\newcommand{\ea}{\end{array}}

\newcommand{\ol}{\overline}

\newcommand{\id}{\mathbb{I}}

\newcommand{\sig}{\sigma}
\newcommand{\Sig}{\Sigma}

\newcommand{\Lam}{\Lambda}
\newcommand{\gam}{\gamma}

\newcommand{\Dta}{\Delta}

\newcommand{\pt}{\partial}

\numberwithin{equation}{section}

\begin{document}
\title[GLM Representation for the 2-NLS equation on the interval]{The GLM representation of the global relation for the two-component nonlinear Schr\"odinger equation on the interval }

\author[J.Xu]{Jian Xu*}
\address{College of Science\\
University of Shanghai for Science and Technology\\
Shanghai 200093\\
People's  Republic of China}
\email{correspondence author: jianxu@usst.edu.cn}

\author[E.Fan]{Engui Fan}
\address{School of Mathematical Sciences, Institute of Mathematics and Key Laboratory of Mathematics for Nonlinear Science\\
Fudan University\\
Shanghai 200433\\
People's  Republic of China}
\email{faneg@fudan.edu.cn}


\date{\today}

\begin{abstract}
In a previous work, we show that the solution of the initial-boundary value problem for the two-component nonlinear Schr\"odinger equation on the finite interval can be expressed in terms of the solution of a $3\times 3$ Riemann-Hilbert problem. The relevant jump matrices are explicitly given in terms of the three matrix-value spectral functions $s(k)$, $S(k)$ and $S_L(k)$, which in turn are defined in terms of the initial values, boundary values at $x=0$ and boundary values at $x=L$, respectively. However, for a well-posed problem, only part of the boundary values can be prescribed, the remaining boundary data cannot be independently specified, but are determined by the so-called global relation. Here, we use a Gelfand-Levitan-Marchenko representation to derive an expression for the generalized Dirichlet-to-Neumann map to characterize the unknown boundary values in physical domain, which is different from the approach, in fact it analyzed the global relation in spectral domain, used in the previous work. And, we can show that these two representations are equivalent.

\end{abstract}

\maketitle

\section{Introduction}

The two-component nonlinear Schr\"odinger equation
\be\label{2-NLSe}
\left\{
\ba{l}
iq_{1t}+q_{1xx}-2\sig(|q_1|^2+|q_2|^2)q_1=0,\\
iq_{2t}+q_{2xx}-2\sig(|q_1|^2+|q_2|^2)q_2=0.
\ea
\right.\qquad \sig=\pm 1.
\ee
where $q_1(x,t)$ and $q_2(x,t)$ are complex-valued functions. Here, $\sig=1$ means defocusing case and $\sig=-1$ means focusing case.
It was first introduced by Manakov to describe the propagation of an optical
pulse in a birefringent optical fiber \cite{m74}, so it is usually called Manokov equation or Manokov system. Subsequently, this system also arises
in the context of multicomponent Bose-Einstein condensates \cite{ba2001}. It is an integrable equation and the initial value problem on the line can be analyzed by means of the Inverse Scattering Transform (IST) as demonstrated by Manokov in \cite{m74}. However, in many (perhaps most) laboratory and field situations, the solution is generated by what corresponds to the imposition of boundary conditions rather than initial conditions. Thus, we need analyze these equations with boundary value problems (BVPs), or initial-boundary value problems (IBVPs), instead of pure initial-value problems. In the last eighteen years, a generalization of the IST developed by Fokas and his collaborators, has made it possible to analyze IBVPs for integrable equations \cite{f1,f2,f3,f4,fi1,fi2,fi3,fis,fi4,abmfs1,abmfs2,fcpam}. Initially, these developments were all carried out for equations with Lax pairs involving $2 \times 2$ matrices. However, in \cite{l1} the methodology was further developed to include the case of equations with $3\times 3$ Lax pairs. Since the work of \cite{l1}, in which IBVPs were considered on the half-line, it has been a natural problem to extend the Fokas methodology to the case of initial-boundary value problems on an interval.

In a previous work \cite{jf3}, the authors show that the solution of the IBVPs on an interval $\Omega=\{(x,t)|0\le x\le L, 0\le t\le T\}$, here $L>0$ is a positive fixed constant and $T>0$ being a fixed final time, for the two-component nonlinear Schr\"odinger equation can be recovered in terms of the solution of a $3\times 3$ matrix Riemann-Hilbert problem. The relevant jump matrices are explicitly given in terms of the three matrix-value spectral functions $s(k)$, $S(k)$ and $S_L(k)$. The matrix function $s(k)$ is defined in terms of the initial data $q_{10}(x)=q_1(x,t=0),q_{20}(x)=q_2(x,t=0)$ via a system of linear Volterra integral equations; the matrix functions $\{S(k),S_L(k)\}$ are defined in terms of the boundary data at $x=0$ and boundary values at $x=L$, respectively, also via systems of linear Volterra integral equations. However, the integral equations defining $\{S(k),S_L(k)\}$ involve {\it all} boundary values, whereas for a well-posed problem, only part of the boundary values can be prescribed, the remaining boundary data cannot be independently specified.  Thus, the complete solution of a concrete IBVPs requires the characterization of $\{S(k),S_L(k)\}$ in terms of the given initial and boundary conditions. Since these three matrix functions $\{s(k),S(k),S_L(k)\}$ are determined by the so-called {\it global relation}, it makes the characterization possible. Thus, before the functions $\{S(k),S_L(k)\}$ can be constructed from the above linear integral equations, the global relation must first be used to eliminate the unknown boundary data.

The analysis of the global relation can take place in two different domains: in the physical domain or in the spectral domain. Although
these two domains are related by a transform, each viewpoint has its own advantages. In the previous work \cite{jf3}, we did the analysis in the spectral domain. In this paper, we do the analysis in the physical domain. And we also show that the expression for the generalized Dirichlet-to-Neumann map (i.e. the map which determines the unknown boundary values from the known ones) in this paper is equivalent to the expression obtained in section 4 in \cite{jf3}.

{\bf Organization of the paper:} In section 2 we recall the Lax pair formulation and the global relation associated with the two-component nonlinear Schr\"odinger equation. In section 3, we derive a GLM representation for an appropriate eigenfunction of the Lax pair. In section 4, we analyze both the Dirichlet and Neumann problems of the two-component nonlinear Schr\"odinger equation on the finite interval with zero initial conditions. Furthermore, in section 4 we establish the equivalence of the formulas obtained in \cite{jf3} with the formulas obtained via the GLM representations.

\section{The global relation}

\subsection{Lax pair and spectral analysis}
The 2-NLS equation admits a $3\times 3$ Lax pair,
\begin{subequations}\label{Laxpair}
\be\label{Lax-x}
\Psi_x=U\Psi,\quad \Psi=\left(\ba{c}\Psi_1\\\Psi_2\\\Psi_3\ea\right).
\ee
\be\label{Lax-t}
\Psi_t=V\Psi.
\ee
\end{subequations}
where
\be\label{Udef}
U=ik\Lam+V_1.
\ee
and
\be\label{Vdef}
V=2ik^2\Lam+V_2
\ee
here
\be\label{Lamdef}
\Lam=\left(\ba{ccc}-1&0&0\\0&1&0\\0&0&1\ea\right),V_1=\left(\ba{ccc}0&q_1&q_2\\ \sig\bar q_1&0&0\\ \sig\bar q_2&0&0\ea\right),V_2=2kV_2^{(1)}+V_2^{(0)}.
\ee
where
\be
V_2^{(1)}=V_1,\qquad
V_2^{(0)}=i\Lam (V^2_1-V_{1x}).
\ee

Following \cite{jf3}, we introduce a new eigenfunction $\mu(x,t,k)$ by
\be\label{neweigfun}
\Psi=\mu e^{i\Lam kx+2i\Lam k^2t}
\ee
then we find the Lax pair equations
\be\label{muLaxe}
\left\{
\ba{l}
\mu_x-[ik\Lam,\mu]=V_1\mu,\\
\mu_t-[2ik^2\Lam,\mu]=V_2\mu.
\ea
\right.
\ee
Letting $\hat A$ denotes the operators which acts on a $3\times 3$ matrix $X$ by $\hat A X=[A,X]$ , then the equations in (\ref{muLaxe}) can be written in differential form as
\be\label{mudiffform}
d(e^{-(ikx+2ik^2t)\hat \Lam}\mu)=W,
\ee
where $W(x,t,k)$ is the closed one-form defined by
\be\label{Wdef}
W=e^{-(ikx+2ik^2t)\hat \Lam}(V_1dx+V_2dt)\mu.
\ee

We introduce four solutions $\{\mu_j(x,t,k)\}_{i=1}^{4}$ of (\ref{muLaxe}) by the Volterra integral equations
\be\label{mujdef}
\mu_j(x,t,k)=\id+\int_{\gam_j}e^{(i kx+2i k^2t)\hat \Lam}W_j(x',t',k).\qquad j=1,2,3,4.
\ee
where $\id$ denote the identity matrix, $W_j$ is given by (\ref{Wdef}) with $\mu$ replaced with $\mu_j$, and the contours $\{\gam_j\}_1^4$ are showed in Figure \ref{fig-1}.
\begin{figure}[th]
\centering
\includegraphics{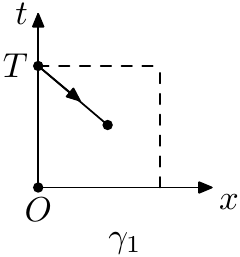}
\includegraphics{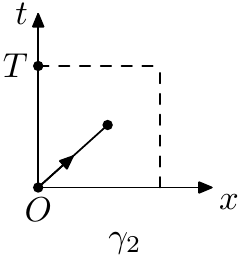}
\includegraphics{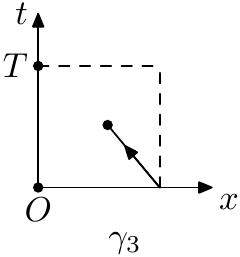}
\includegraphics{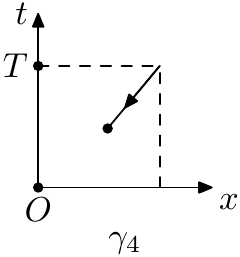}
\caption{The four contours $\gam_1,\gam_2,\gam_3$ and $\gam_4$ in the $(x,t)-$domain.}\label{fig-1}
\end{figure}

The solutions $\{\mu_j(x,t,k)\}_{i=1}^{4}$ satisfy the symmetry
\be\label{symm}
\mu_j^{-1}(x,t,k)=A\ol{\mu_j(x,t,\bar k)}^{T}A,
\ee
where
\be
A=\left(
\ba{ccc}
1&0&0\\
0&-\sig&0\\
0&0&-\sig
\ea
\right),\quad \sig^2=1.
\ee
Here, the superscript $T$ denotes a matrix transpose.

Then, we can define the $3\times 3$ matrix value spectral functions $s(k)$, $S(k)$ and $S_L(k)$ by
\begin{subequations}\label{sSSLdef}
\be\label{smu3}
s(k)=\mu_3(0,0,k),
\ee
\be\label{Smu1}
S(k)=\mu_1(0,0,k)=e^{-2ik^2 T\hat \Lam}\mu_2^{-1}(0,T,k),
\ee
\be\label{SLmu4}
S_L(k)=\mu_4(L,0,k)=e^{-2ik^2 T\hat \Lam}\mu_3^{-1}(L,T,k).
\ee
\end{subequations}

We also introduce the functions $\{\Phi_{ij}(t,k),\phi_{ij}(t,k)\}_{i,j=1}^{3}$ as follows
\begin{subequations}
\be\label{Phidef}
\mu_2(0,t,k)=\left(\ba{lll}\Phi_{11}(t,k)&\Phi_{12}(t,k)&\Phi_{13}(t,k)\\
\Phi_{21}(t,k)&\Phi_{22}(t,k)&\Phi_{23}(t,k)\\\Phi_{31}(t,k)&\Phi_{32}(t,k)&\Phi_{33}(t,k)\ea\right),
\ee
\be\label{phidef}
\mu_3(L,t,k)=\left(\ba{lll}\phi_{11}(t,k)&\phi_{12}(t,k)&\phi_{13}(t,k)\\
\phi_{21}(t,k)&\phi_{22}(t,k)&\phi_{23}(t,k)\\\phi_{31}(t,k)&\phi_{32}(t,k)&\phi_{33}(t,k)\ea\right).
\ee
\end{subequations}

Denoting the sets $\{D_j\}_{j=1}^{4}$ by (see Figure \ref{fig-2}),
\[
D_j=\{\frac{j-1}{2}\pi<\arg{k}<\frac{j}{2}\pi\},
\]
\begin{figure}[th]
\centering
\includegraphics{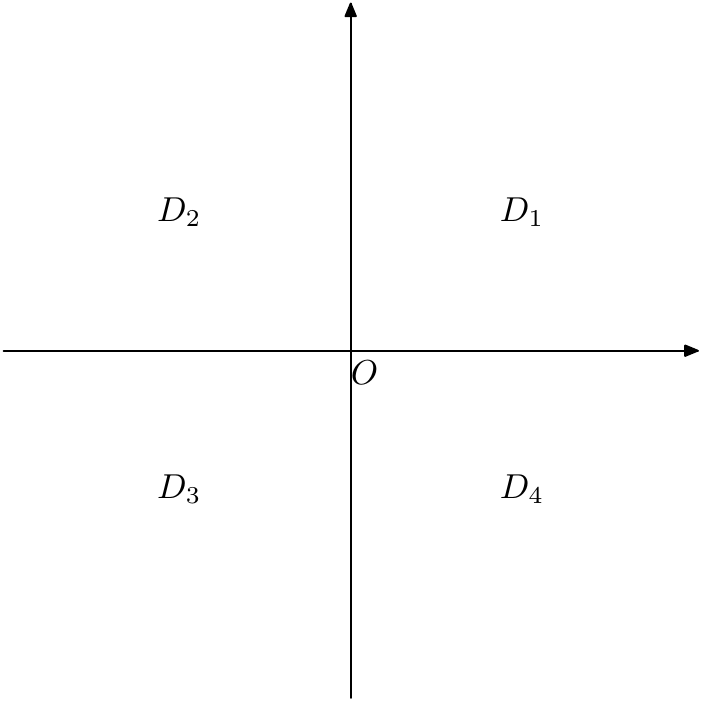}
\caption{The sets $D_n$, $n=1,\ldots ,4$, which decompose the complex $k-$plane.}\label{fig-2}
\end{figure}
it follows from (\ref{mujdef}) and (\ref{sSSLdef}), the functions $\{s(k),S(k),S_L(k)\}$ have the following boundedness properties:
\[
\ba{ll}
s(k):&(D_3\cup D_4,D_1\cup D_2,D_1\cup D_2),\\
S(k):&(D_2\cup D_4,D_1\cup D_3,D_1\cup D_3),\\
S_L(k):&(D_2\cup D_4,D_1\cup D_3,D_1\cup D_3)\\
\ea
\]

\subsection{The global relation}
The spectral functions $S(k),S_L(k)$ and $s(k)$ are not independent but satisfy an important relation. Indeed, it follows from (\ref{sSSLdef}) that
\be
\mu_1(x,t,k)e^{(ikx+2ik^2t)\hat \Lam}\{S^{-1}(k)s(k)e^{-ikL\hat\Lam}S_L(k)\}=\mu_4(x,t,k).
\ee
Since $\mu_1(0,T,k)=\id$, evaluation at $(0,T)$ yields the following global relation:
\be\label{globalrel}
S^{-1}(k)s(k)e^{-ikL\hat\Lam}S_L(k)=e^{-2ik^2T\hat \Lam}c(T,k),
\ee
where $c(T,k)=\mu_4(0,T,k)$.
It imposes a relation between the Dirichlet and Neumann boundary values of $q_{1}(x,t)$ and $q_{2}(x,t)$. The Dirichlet-to-
Neumann map is determined by solving this relation for the unknown boundary values.

\section{The GLM approach}

In this section, we derive Gelfand-Levitan-Marchenko (GLM) representations for the eigenfunctions $\Phi_{ij}$ and $\phi_{ij}$.


\begin{theorem}
The eigenfunctions $\{\Phi_{ij}\}_{i,j=1}^{3}$ admit the following GLM representations,
\begin{subequations}\label{PhiGLM}
\be\label{Phi1GLM}
\footnotesize
\left\{
\ba{l}
\Phi_{11}(t,k)=1+\int_{-t}^{t}\left[\tilde L_{11}(t,s)-\frac{i}{2}(g_{01}(t)M_{21}(t,s)+g_{02}(t)M_{31}(t,s))+kM_{11}(t,s)\right]e^{-2ik^2(s-t)}ds\\
\Phi_{21}(t,k)=\int_{-t}^{t}\left[\tilde L_{21}(t,s)+\frac{i}{2}\sig \bar g_{01}(t)M_{11}(t,s)+kM_{21}(t,s)\right]e^{-2ik^2(s-t)}ds\\
\Phi_{31}(t,k)=\int_{-t}^{t}\left[\tilde L_{31}(t,s)+\frac{i}{2}\sig \bar g_{02}(t)M_{11}(t,s)+kM_{31}(t,s)\right]e^{-2ik^2(s-t)}ds.
\ea
\right.
\ee
\be\label{Phi2GLM}
\footnotesize
\left\{
\ba{l}
\Phi_{12}(t,k)=\int_{-t}^{t}\left[\tilde L_{12}(t,s)-\frac{i}{2}(g_{01}(t)M_{22}(t,s)+g_{02}(t)M_{32}(t,s))+kM_{12}(t,s)\right]e^{2ik^2(s-t)}ds\\
\Phi_{22}(t,k)=1+\int_{-t}^{t}\left[\tilde L_{22}(t,s)+\frac{i}{2}\sig \bar g_{01}(t)M_{12}(t,s)+kM_{22}(t,s)\right]e^{2ik^2(s-t)}ds\\
\Phi_{32}(t,k)=\int_{-t}^{t}\left[\tilde L_{32}(t,s)+\frac{i}{2}\sig \bar g_{02}(t)M_{12}(t,s)+kM_{32}(t,s)\right]e^{2ik^2(s-t)}ds.
\ea
\right.
\ee
\be\label{Phi3GLM}
\footnotesize
\left\{
\ba{l}
\Phi_{13}(t,k)=\int_{-t}^{t}\left[\tilde L_{13}(t,s)-\frac{i}{2}(g_{01}(t)M_{23}(t,s)+g_{02}(t)M_{33}(t,s))+kM_{13}(t,s)\right]e^{2ik^2(s-t)}ds\\
\Phi_{23}(t,k)=\int_{-t}^{t}\left[\tilde L_{23}(t,s)+\frac{i}{2}\sig \bar g_{01}(t)M_{13}(t,s)+kM_{23}(t,s)\right]e^{2ik^2(s-t)}ds\\
\Phi_{33}(t,k)=1+\int_{-t}^{t}\left[\tilde L_{33}(t,s)+\frac{i}{2}\sig \bar g_{02}(t)M_{13}(t,s)+kM_{33}(t,s)\right]e^{2ik^2(s-t)}ds.
\ea
\right.
\ee
\end{subequations}
Here the functions $\{M_{ij}(t,s),\tilde L_{ij}(t,s)\}_{i,j=1}^{3}$ are the elements of the $3\times 3$ matrix $\tilde L(t,s)$ and $M(t,s)$. And they satisfy the initial conditions
\begin{subequations}\label{Goursatintdata}
\footnotesize
\be
\left\{
\ba{l}
M_{11}(t,-t)=M_{22}(t,-t)=M_{23}(t,-t)=M_{32}(t,-t)=M_{33}(t,-t)=0\\
\tilde L_{11}(t,-t)=\tilde L_{22}(t,-t)=\tilde L_{23}(t,-t)=\tilde L_{32}(t,-t)=\tilde L_{33}(t,-t)=0.
\ea
\right.
\ee

\be
\left\{
\ba{llll}
M_{12}(t,t)=g_{01}(t)& M_{13}(t,t)=g_{02}(t)& M_{21}(t,t)=\sig \bar g_{01}(t)& M_{31}(t,t)=\sig \bar g_{02}(t)\\
\tilde L_{12}(t,t)=\frac{i}{2}g_{11}(t)& \tilde L_{13}(t,t)=\frac{i}{2}g_{12}(t)& \tilde L_{21}(t,t)=-\frac{i}{2}\sig \bar g_{11}(t)& \tilde L_{31}(t,t)=-\frac{i}{2}\sig \bar g_{12}(t)
\ea
\right.
\ee
\end{subequations}

and an ODE systems
\begin{subequations}\label{MGoursat}
\be
\footnotesize
\left\{
\ba{l}
M_{11,t}(t,s)+M_{11,s}(t,s)=2(g_{01}(t)\tilde L_{21}(t,s)+g_{02}(t)\tilde L_{31}(t,s))+i(g_{11}(t)M_{21}(t,s)+g_{12}(t)M_{31}(t,s))\\
M_{21,t}(t,s)-M_{21,s}(t,s)=2\sig \bar g_{01}(t)\tilde L_{11}(t,s)-i\sig \bar g_{01}(t)M_{11}(t,s)\\
M_{31,t}(t,s)-M_{31,s}(t,s)=2\sig \bar g_{02}(t)\tilde L_{11}(t,s)-i\sig \bar g_{02}(t)M_{11}(t,s)
\ea
\right.
\ee
\be
\footnotesize
\left\{
\ba{l}
M_{12,t}(t,s)-M_{12,s}(t,s)=2(g_{01}(t)\tilde L_{22}(t,s)+g_{02}(t)\tilde L_{32}(t,s))+i(g_{11}(t)M_{22}(t,s)+g_{12}(t)M_{32}(t,s))\\
M_{22,t}(t,s)+M_{22,s}(t,s)=2\sig \bar g_{01}(t)\tilde L_{12}(t,s)-i\sig \bar g_{01}(t)M_{12}(t,s)\\
M_{32,t}(t,s)+M_{32,s}(t,s)=2\sig \bar g_{02}(t)\tilde L_{12}(t,s)-i\sig \bar g_{02}(t)M_{12}(t,s)
\ea
\right.
\ee
\be
\footnotesize
\left\{
\ba{l}
M_{13,t}(t,s)-M_{13,s}(t,s)=2(g_{01}(t)\tilde L_{23}(t,s)+g_{02}(t)\tilde L_{33}(t,s))+i(g_{11}(t)M_{23}(t,s)+g_{12}(t)M_{33}(t,s))\\
M_{23,t}(t,s)+M_{23,s}(t,s)=2\sig \bar g_{01}(t)\tilde L_{13}(t,s)-i\sig \bar g_{01}(t)M_{13}(t,s)\\
M_{33,t}(t,s)+M_{33,s}(t,s)=2\sig \bar g_{02}(t)\tilde L_{13}(t,s)-i\sig \bar g_{02}(t)M_{13}(t,s)
\ea
\right.
\ee
\end{subequations}

\begin{subequations}\label{LGoursat}
\scriptsize
\be
\left\{
\ba{rl}
\tilde L_{11,t}(t,s)+\tilde L_{11,s}(t,s)=& i(g_{11}(t)\tilde L_{21}(t,s)+g_{12}(t)\tilde L_{31}(t,s))+\frac{\sig}{2}\left(g_{01}\bar g_{11}+g_{02}\bar g_{12}-\bar g_{01}g_{11}-\bar g_{02}g_{12}\right)(t)M_{11}(t,s)\\
&+\frac{1}{2}\left\{[i\dot{g}_{01}-\sig(|g_{01}|^2+|g_{02}|^2)g_{01}](t)M_{21}(t,s)+[i\dot{g}_{02}-\sig(|g_{01}|^2+|g_{02}|^2)g_{02}](t)M_{31}(t,s)\right\}\\
\tilde L_{21,t}(t,s)-\tilde L_{21,s}(t,s)=& -i\sig\bar g_{11}(t)\tilde L_{11}(t,s)+\frac{\sig}{2}(\bar g_{01}g_{11}-g_{01}\bar g_{11})(t)M_{21}(t,s)+\frac{\sig}{2}(\bar g_{01}g_{12}-g_{02}\bar g_{11})(t)M_{31}(t,s)\\
&+\frac{\sig}{2}\left[-i\dot{\bar g}_{01}-\sig (|g_{01}|^2+|g_{02}|^2)\bar g_{01}\right](t)M_{11}(t,s)\\
\tilde L_{31,t}(t,s)-\tilde L_{31,s}(t,s)=& -i\sig\bar g_{12}(t)\tilde L_{11}(t,s)+\frac{\sig}{2}(\bar g_{02}g_{11}-g_{01}\bar g_{12})(t)M_{21}(t,s)+\frac{\sig}{2}(\bar g_{02}g_{12}-g_{02}\bar g_{12})(t)M_{31}(t,s)\\
&+\frac{\sig}{2}\left[-i\dot{\bar g}_{02}-\sig (|g_{01}|^2+|g_{02}|^2)\bar g_{02}\right](t)M_{11}(t,s)
\ea
\right.
\ee

\be
\left\{
\ba{rl}
\tilde L_{12,t}(t,s)-\tilde L_{12,s}(t,s)=& i(g_{11}(t)\tilde L_{22}(t,s)+g_{12}(t)\tilde L_{32}(t,s))+\frac{\sig}{2}\left(g_{01}\bar g_{11}+g_{02}\bar g_{12}-\bar g_{01}g_{11}-\bar g_{02}g_{12}\right)(t)M_{12}(t,s)\\
&+\frac{1}{2}\left\{[i\dot{g}_{01}-\sig(|g_{01}|^2+|g_{02}|^2)g_{01}](t)M_{22}(t,s)+[i\dot{g}_{02}-\sig(|g_{01}|^2+|g_{02}|^2)g_{02}](t)M_{32}(t,s)\right\}\\
\tilde L_{22,t}(t,s)+\tilde L_{22,s}(t,s)=& -i\sig\bar g_{11}(t)\tilde L_{12}(t,s)+\frac{\sig}{2}(\bar g_{01}g_{11}-g_{01}\bar g_{11})(t)M_{22}(t,s)+\frac{\sig}{2}(\bar g_{01}g_{12}-g_{02}\bar g_{11})(t)M_{32}(t,s)\\
&+\frac{\sig}{2}\left[-i\dot{\bar g}_{01}-\sig (|g_{01}|^2+|g_{02}|^2)\bar g_{01}\right](t)M_{11}(t,s)\\
\tilde L_{32,t}(t,s)+\tilde L_{32,s}(t,s)=& -i\sig\bar g_{12}(t)\tilde L_{12}(t,s)+\frac{\sig}{2}(\bar g_{02}g_{11}-g_{01}\bar g_{12})(t)M_{22}(t,s)+\frac{\sig}{2}(\bar g_{02}g_{12}-g_{02}\bar g_{12})(t)M_{32}(t,s)\\
&+\frac{\sig}{2}\left[-i\dot{\bar g}_{02}-\sig (|g_{01}|^2+|g_{02}|^2)\bar g_{02}\right](t)M_{12}(t,s)
\ea
\right.
\ee

\be
\left\{
\ba{rl}
\tilde L_{13,t}(t,s)-\tilde L_{13,s}(t,s)=& i(g_{11}(t)\tilde L_{23}(t,s)+g_{12}(t)\tilde L_{33}(t,s))+\frac{\sig}{2}\left(g_{01}\bar g_{11}+g_{02}\bar g_{12}-\bar g_{01}g_{11}-\bar g_{02}g_{12}\right)(t)M_{13}(t,s)\\
&+\frac{1}{2}\left\{[i\dot{g}_{01}-\sig(|g_{01}|^2+|g_{02}|^2)g_{01}](t)M_{23}(t,s)+[i\dot{g}_{02}-\sig(|g_{01}|^2+|g_{02}|^2)g_{02}](t)M_{33}(t,s)\right\}\\
\tilde L_{23,t}(t,s)+\tilde L_{23,s}(t,s)=& -i\sig\bar g_{11}(t)\tilde L_{13}(t,s)+\frac{\sig}{2}(\bar g_{01}g_{11}-g_{01}\bar g_{11})(t)M_{23}(t,s)+\frac{\sig}{2}(\bar g_{01}g_{12}-g_{02}\bar g_{11})(t)M_{33}(t,s)\\
&+\frac{\sig}{2}\left[-i\dot{\bar g}_{01}-\sig (|g_{01}|^2+|g_{02}|^2)\bar g_{01}\right](t)M_{13}(t,s)\\
\tilde L_{33,t}(t,s)+\tilde L_{33,s}(t,s)=& -i\sig\bar g_{12}(t)\tilde L_{13}(t,s)+\frac{\sig}{2}(\bar g_{02}g_{11}-g_{01}\bar g_{12})(t)M_{23}(t,s)+\frac{\sig}{2}(\bar g_{02}g_{12}-g_{02}\bar g_{12})(t)M_{33}(t,s)\\
&+\frac{\sig}{2}\left[-i\dot{\bar g}_{02}-\sig (|g_{01}|^2+|g_{02}|^2)\bar g_{02}\right](t)M_{13}(t,s)
\ea
\right.
\ee
\end{subequations}

\normalsize
The GLM representations of $\{\phi_{ij}\}_{i,j=1}^{3}$ are similar to equations (\ref{PhiGLM}). There are two differences. First, replacing the boundary data $\{g_{01}(t),g_{02}(t),g_{11}(t),g_{12}(t)\}$ by $\{f_{01}(t),f_{02}(t),f_{11}(t),f_{12}(t)\}$. Second, replacing the functions $M_{ij}(t,s), \tilde L_{ij}(t,s)$ by $\mathcal{M}_{ij}(t,s), \mathcal{\tilde L}_{ij}(t,s)$. And the functions $\mathcal{M}_{ij}(t,s), \mathcal{\tilde L}_{ij}(t,s)$ satisfy the similar systems of equations (\ref{MGoursat}) and (\ref{LGoursat}) with $\{g_{01}(t),g_{02}(t),g_{11}(t),g_{12}(t)\}$ replaced by $\{f_{01}(t),f_{02}(t),f_{11}(t),f_{12}(t)\}$, too.
\end{theorem}

\begin{proof}

Assume that
\be\label{PsiGlm}
\Psi(t,k)=e^{2ik^2\Lam t}+\int_{-t}^{t}\left(L(t,s)+k M(t,s)\right)e^{2ik^2s\Lam} ds,
\ee
where $L$ and $M$ are $3\times 3$ matrices. Substituting the above equation (\ref{PsiGlm}) into the $t-part$ of the Lax pair (\ref{Laxpair}) with the boundary condition $\Psi(0,k)=\id$, and noticing that the identity
\be
2ik^2 \int_{-t}^{t}F(t,s)e^{2ik^2s\Lam}ds=\left[F(t,t)e^{2ik^2t\Lam}-F(t,-t)e^{-2ik^2t\Lam}-\int_{-t}^{t}F_s(t,s)e^{2ik^2s\Lam}ds\right]\Lam
\ee
where $F(t,s)$ is a $3\times 3-$matrix-value function (this identity is derived directly by using integration by parts),
we find the following equations:
\begin{subequations}
\be
M(t,-t)+\Lam M(t,-t)\Lam =0.
\ee
\be
L(t,-t)+\Lam L(t,-t)\Lam -iV^{(1)}_2 M(t,-t)\Lam=0.
\ee
\end{subequations}

\begin{subequations}
\be
M(t,t)-\Lam M(t,t)\Lam =2V^{(1)}_2.
\ee
\be
L(t,t)-\Lam L(t,t)\Lam +iV^{(1)}_2 M(t,t)\Lam=V^{(0)}_2.
\ee
\end{subequations}

\begin{subequations}
\be
M_t(t,s)+\Lam M_s(t,s)\Lam =2V^{(1)}_2 L(t,s)+V^{(0)}_2 M(t,s).
\ee
\be
L_t(t,s)+\Lam L_s(t,s)\Lam =iV^{(1)}_2 M_s(t,s)\Lam+V^{(0)}_2 L(t,s).
\ee
\end{subequations}

Set
\be\label{tildeLdef}
L(t,s)=\tilde L(t,s)-\frac{i}{2}V^{(1)}_2 \Lam M(t,s).
\ee

Then, we find

\begin{subequations}\label{MtildeL-tint}
\be
M(t,-t)+\Lam M(t,-t)\Lam =0.
\ee
\be
\tilde L(t,-t)+\Lam \tilde L(t,-t)\Lam=0.
\ee
\end{subequations}

\begin{subequations}\label{MtildeLtint}
\be
M(t,t)-\Lam M(t,t)\Lam =2V^{(1)}_2.
\ee
\be
\tilde L(t,t)-\Lam \tilde L(t,t)\Lam= -i\Lam V^{(1)}_{2x}.
\ee
\end{subequations}

\begin{subequations}\label{MtildeLequs}
\be
M_t(t,s)+\Lam M_s(t,s)\Lam =2V^{(1)}_2 \tilde L(t,s)-(i(V^{(1)}_2)^2\Lam-V^{(0)}_2) M(t,s).
\ee
\be
\ba{rl}
\tilde L_t(t,s)+\Lam \tilde L_s(t,s)\Lam =&(V^{(0)}_2+iV^{(1)}_2\Lam V^{(1)}_2)\tilde L(t,s)\\
&+\frac{i}{2}(\dot{V^{(1)}_2} \Lam-iV^{(1)}_2\Lam (V^{(1)}_2)^2\Lam +V^{(1)}_2 \Lam V^{(0)}_2-V^{(0)}_2 V^{(1)}_2 \Lam)M(t,s).
\ea
\ee
where the $\dot{f}$ denotes that $\frac{df}{dt}$.
\end{subequations}

Recalling that the definition of $V^{(0)}_2=i\Lam((V^{(1)}_2)^2-V^{(1)}_{2x})$, we can write the equation (\ref{MtildeLequs}) as
\begin{subequations}\label{MtildeLequs'}
\be
M_t(t,s)+\Lam M_s(t,s)\Lam =2V^{(1)}_2 \tilde L(t,s)-i\Lam V^{(1)}_{2x} M(t,s).
\ee
\be
\ba{rl}
\tilde L_t(t,s)+\Lam \tilde L_s(t,s)\Lam =&-i\Lam V^{(1)}_{2x}\tilde L(t,s)\\
&+\frac{i}{2}(\dot{V^{(1)}_2} \Lam-i(V^{(1)}_2)^3+iV^{(1)}_{2x}V^{(1)}_2-iV^{(1)}_2 V^{(1)}_{2x})M(t,s).
\ea
\ee
\end{subequations}
If we denote the matrices $M(t,s)$ and $\tilde L(t,s)$ by
\be
M(t,s)=\left(\ba{ccc}M_{11}(t,s)&M_{12}(t,s)&M_{13}(t,s)\\M_{21}(t,s)&M_{22}(t,s)&M_{23}(t,s)\\M_{31}(t,s)&M_{32}(t,s)&M_{33}(t,s)\ea\right),
\quad
\tilde L(t,s)=\left(\ba{ccc}\tilde L_{11}(t,s)&\tilde L_{12}(t,s)&\tilde L_{13}(t,s)\\\tilde L_{21}(t,s)&\tilde L_{22}(t,s)&\tilde L_{23}(t,s)\\\tilde L_{31}(t,s)&\tilde L_{32}(t,s)&\tilde L_{33}(t,s)\ea\right)
\ee
By the equations (\ref{MtildeL-tint}), (\ref{MtildeLtint}) and (\ref{MtildeLequs'}), we can get the initial conditions (\ref{Goursatintdata}) and the ODE systems (\ref{MGoursat}), (\ref{LGoursat}), respectively.

Finally, noticing that the relation $\mu(0,t,k)=\Psi(t,k)e^{2ik^2t\Lam}$ and the definition (\ref{Phidef}) of $\{\Phi_{ij}(t,k)\}_{i,j=1}^{3}$, we can get the GLM representations (\ref{PhiGLM}) from the equations (\ref{PsiGlm}) and (\ref{tildeLdef}). Similarly to prove the results for $\{\phi_{ij}(t,k)\}_{i,j=1}^{3}$.
\end{proof}

\section{The solution of the global relation}

In this section, we consider the solution of the global relation. Theorems \ref{maintheorem} below leads to expressions
for the Dirichlet-to-Neumann map in terms of the GLM representations.

\subsection{The Dirichlet-to-Neumann Map}

In order to simplify our formulas, we use some notations as following:

\begin{itemize}
\item
For a function $f(t,s)$, we let $\hat f(t,k)$ denote the transform
\be
\hat f(t,k)=\int_{-t}^{t}f(t,s)e^{2ik^2(s-t)}ds.
\ee
\item The functions $f_+(k)$ and $f_-(k)$ denote the following even and odd combinations of the function $f(k)$ :
\[
f_+(k)=f(k)+f(-k),\quad f_-(k)=f(k)-f(-k),\quad k\in \C.
\]

\item $\Dta(k)$ and $\Sig(k)$ are defined by
\[
\Dta(k)=e^{2ikL}-e^{-2ikL},\quad \Sig (k)=e^{2ikL}+e^{-2ikL}.
\]
\end{itemize}
If we partition the $3\times 3$ matrix $A=(A_{ij})_{i,j=1}^{3}$ by $A=\left(\ba{cc}A_{11}&A_{1j}\\A_{j1}&A_{2\times 2}\ea\right),j=2,3$, and denote $g_0$ and $f_0$ as two component row vectors by $g_0=\left(\ba{cc}g_{01}(t)&g_{02}(t)\ea\right)$ and $f_0=\left(\ba{cc}f_{01}(t)&f_{02}(t)\ea\right)$, respectively, then the GLM representations of $\{\Phi_{ij},\phi_{ij}\}$ can be written as
\begin{subequations}\label{GLMFourier}
\be
\left\{
\ba{ll}
\Phi_{11}=1+\ol{\widehat{\bar{\tilde L}}}_{11}-\frac{i}{2}g_{0}\ol{\widehat{\bar M}}_{j1}+k\ol{\widehat{\bar M}}_{11},&
\Phi_{1j}=\widehat{\tilde L}_{1j}-\frac{i}{2}g_{0}\hat M_{2\times 2}+k\hat M_{1j}\\
\Phi_{j1}=\ol{\widehat{\bar{\tilde L}}}_{j1}+\frac{i}{2}\sig \bar g^T_{0}\ol{\widehat{\bar M}}_{11}+k\ol{\widehat{\bar M}}_{j1},&
\Phi_{2\times 2}=\id+\widehat{\tilde L}_{2\times 2}+\frac{i}{2}\sig \bar g^T_{0}\hat M_{1j}+k\hat M_{2\times 2}.\\
\ea
\right.
\ee

\be
\left\{
\ba{ll}
\phi_{11}=1+\ol{\widehat{\bar{\tilde{\mathcal L}}}}_{11}-\frac{i}{2}f_{0}\ol{\widehat{\bar{\mathcal{M}}}}_{j1}+k\ol{\widehat{\bar{\mathcal{M}}}}_{11},&
\phi_{1j}=\widehat{\tilde{\mathcal{L}}}_{1j}-\frac{i}{2}f_{0}\hat{\mathcal{M}}_{2\times 2}+k\hat{\mathcal{M}}_{1j}\\
\phi_{j1}=\ol{\widehat{\bar{\tilde{\mathcal{L}}}}}_{j1}+\frac{i}{2}\sig \bar f^T_{0}\ol{\widehat{\bar{\mathcal{M}}}}_{11}+k\ol{\widehat{\bar{\mathcal{M}}}}_{j1},&
\phi_{2\times 2}=\id+\widehat{\tilde{\mathcal{L}}}_{2\times 2}+\frac{i}{2}\sig \bar f^T_{0}\hat{\mathcal{M}}_{1j}+k\hat{\mathcal{M}}_{2\times 2}.\\
\ea
\right.
\ee
\end{subequations}
where $\ol{\widehat{\bar M}}$ is short-hand notation for $\ol{\widehat{\bar M(t,\bar k)}}$ etc.

\begin{theorem}\label{maintheorem}
Let $T<\infty$. Let $q_{10}(x)=q_{20}(x)=0,0\le x\le L$, be two vanishing initial data.
\par
For the Dirichlet problem it is assumed that the functions $\{g_{01}(t),g_{02}(t)\},0\le t<T$, have sufficient smoothness and are compatible with $\{q_{10}(x),q_{20}(x)\}$ at $x=t=0$, that is
\[
q_{10}(0)=g_{01}(0),\quad q_{20}(0)=g_{02}(0).
\]
The functions $\{f_{01}(t),f_{02}(t)\},0\le t<T$, have sufficient smoothness and are compatible with $\{q_{10}(x),q_{20}(x)\}$ at $x=L$, that is,
\[
q_{10}(L)=f_{01}(0),\quad q_{20}(L)=f_{02}(0).
\]
\par
For the Neumann problem it is assumed that the functions $\{g_{11}(t),g_{12}(t)\},0\le t<T$, have sufficient smoothness and are compatible with $\{q_{10}(x),q_{20}(x)\}$ at $x=t=0$; the functions $\{f_{11}(t),f_{12}(t)\},0\le t<T$, have sufficient smoothness and are compatible with $\{q_{10}(x),q_{20}(x)\}$ at $x=L$.
\par
Then the spectral functions $\{S(k),S_L(k)\}$ are given by
\begin{subequations}\label{SSLK}
\be\label{SK}
S(k)=\left(\ba{ccc}
\ol{\Phi_{11}(\bar k)}&-\sig \ol{\Phi_{21}(\bar k)}e^{4ik^2T}&-\sig \ol{\Phi_{31}(\bar k)}e^{4ik^2T}\\
-\sig \ol{\Phi_{12}(\bar k)}e^{-4ik^2T}&\ol{\Phi_{22}(\bar k)}&\ol{\Phi_{32}(\bar k)}\\
-\sig \ol{\Phi_{13}(\bar k)}e^{-4ik^2T}&\ol{\Phi_{23}(\bar k)}&\ol{\Phi_{33}(\bar k)}\ea\right)
\ee
\be\label{SLk}
S_L(k)=\left(\ba{ccc}
\ol{\phi_{11}(\bar k)}&-\sig \ol{\phi_{21}(\bar k)}e^{4ik^2T}&-\sig \ol{\phi_{31}(\bar k)}e^{4ik^2T}\\
-\sig \ol{\phi_{12}(\bar k)}e^{-4ik^2T}&\ol{\phi_{22}(\bar k)}&\ol{\phi_{32}(\bar k)}\\
-\sig \ol{\phi_{13}(\bar k)}e^{-4ik^2T}&\ol{\phi_{23}(\bar k)}&\ol{\phi_{33}(\bar k)}\ea\right)
\ee
\end{subequations}
and the complex-value functions $\{\Phi_{ij}(t,k)\}_{i,j=1}^{3}$ satisfy the GLM representations defined as (\ref{PhiGLM}), the complex-value functions $\phi_{ij}(t,k)\}_{i,j=1}^{3}$ satisfy the similar GLM representations as (\ref{PhiGLM}).

Define the functions $\{F_{1j}(t,k),\mathcal{F}_{1j}(t,k)\},j=2,3$ by
\begin{subequations}
\be\label{F1jdef}
\ba{rl}
F_{1j}=&-\frac{i}{2}g_0 \hat M_{2\times 2}+\frac{i}{2}f_0\widehat{\bar{\mathcal{M}}}_{11}e^{2ikL}\\
&{}+\left(\widehat{\tilde L}_{1j}-\frac{i}{2}g_{0}\hat M_{2\times 2}+k\hat M_{1j}\right)\left(\ol{\widehat{\tilde{\mathcal{L}}}}^{T}_{2\times 2}-\frac{i}{2}\sig \ol{\hat{\mathcal{M}}}^{T}_{1j}f_{0}+k\ol{\hat{\mathcal{M}}}^{T}_{2\times 2}\right)\\
&{}-\sig \left(\ol{\widehat{\bar{\tilde L}}}_{11}-\frac{i}{2}g_{0}\ol{\widehat{\bar M}}_{j1}+k\ol{\widehat{\bar M}}_{11}\right) \left(\widehat{\bar{\tilde{\mathcal{L}}}}^{T}_{j1}-\frac{i}{2}\sig f_{0}\widehat{\bar{\mathcal{M}}}_{11}+k\widehat{\bar{\mathcal{M}}}^{T}_{j1}\right)
\ea
\ee
\be\label{F1jdefmathcal}
\ba{rl}
\mathcal{F}_{1j}=&\frac{i}{2}f_0 \widehat{\mathcal{M}}_{2\times 2}-\frac{i}{2}g_0\hat{\bar{M}}_{11}e^{-2ikL}\\
&-\left(\widehat{\tilde{\mathcal{L}}}_{1j}-\frac{i}{2}f_{0}\hat{\mathcal{M}}_{2\times 2}+k\hat{\mathcal{M}}_{1j}\right)\left(\ol{\widehat{\tilde L}}^{T}_{2\times 2}-\frac{i}{2}\sig \ol{\hat M}^{T}_{1j}g_{0}+k\ol{\hat M}^{T}_{2\times 2}\right)\\
&{}+\sig \left(\ol{\widehat{\bar{\tilde{\mathcal L}}}}_{11}-\frac{i}{2}f_{0}\ol{\widehat{\bar{\mathcal{M}}}}_{j1}+k\ol{\widehat{\bar{\mathcal{M}}}}_{11}\right)\left(\widehat{\bar{\tilde L}}^{T}_{j1}-\frac{i}{2}\sig g_{0}\widehat{\bar M}_{11}+k\widehat{\bar M}^{T}_{j1}\right)e^{-2ikL}
\ea
\ee
\end{subequations}
Under the vanishing intial value assumption, the following formulas are valid:

\begin{enumerate}
  \item For the Dirichlet problem, the unknown boundary values $g_1=\left(g_{11}(t)\quad g_{12}(t)\right)$ and $f_1=\left(f_{11}(t)\quad f_{12}(t)\right)$ are given by

      \begin{subequations}
      \be\label{g1rep}
      \ba{rcl}
      g_{1}&=&\frac{4}{i\pi}\int_{\pt D^{0}_1}\left\{\frac{k^2 \Sig}{\Dta}\left[\hat M_{1j}-\frac{g_0}{2ik^2}\right]-\frac{2k^2\sig}{\Dta}\left[\widehat{\bar{\mathcal{M}}}^{T}_{j1}-\frac{\sig g_0}{2ik^2}\right]+\frac{ikg_0}{2}\hat M_{2\times 2}\right.\\
      &&{}\left. +\frac{kg_0}{2}\ol{\hat{\mathcal{M}}}^{T}_{2\times 2}+\frac{k}{\Dta}\left(e^{-2ikL}F_{1j}\right)_-\right\}dk
      \ea
      \ee
      \be\label{f1rep}
      \ba{rcl}
      f_{1}&=&\frac{4}{i\pi}\int_{\pt D^{0}_1}\left\{-\frac{k^2 \Sig}{\Dta}\left[\widehat{\mathcal{M}}_{1j}-\frac{f_0}{2ik^2}\right]+\frac{2k^2\sig}{\Dta}\left[\widehat{\bar{M}}^{T}_{j1}-\frac{\sig f_0}{2ik^2}\right]+\frac{ikf_0}{2}\hat{\mathcal{M}}_{2\times 2}\right.\\
      &&{}\left. +\frac{kf_0}{2}\ol{\hat{M}}^{T}_{2\times 2}+\frac{k}{\Dta}\left(e^{2ikL}\mathcal{F}_{1j}\right)_-\right\}dk
      \ea
      \ee
      \end{subequations}

  \item For the Neumann problem, the unknown boundary values $g_0=\left(g_{01}(t)\quad g_{02}(t)\right)$ and $f_0=\left(f_{01}(t)\quad f_{02}(t)\right)$ are given by

      \begin{subequations}\label{g0f0rep}
      \be\label{g0rep}
      g_0=\frac{2}{\pi}\int_{\pt D^{0}_1}\frac{1}{\Dta}\left\{\Sig \hat{\tilde L}_{1j}-2\widehat{\tilde{\mathcal{L}}}_{1j}+\left(e^{-2ikL}F_{1j}\right)_+\right\}dk
      \ee
      \be\label{f0rep}
      f_0=\frac{2}{\pi}\int_{\pt D^{0}_1}\frac{1}{\Dta}\left\{-\Sig \widehat{\tilde{\mathcal{L}}}_{1j}+2\sig\hat{\tilde L}_{1j}+\left(e^{2ikL}\mathcal{F}_{1j}\right)_+\right\}dk
      \ee

      \end{subequations}

\end{enumerate}
\end{theorem}

\begin{proof}
The expressions (\ref{SSLK}) of $S(k)$ and $S_L(k)$ are similarly proved as \cite{jf3}. Let us first consider the Dirichlet problem to prove the equation (\ref{g1rep}). Noticing that the global relation (\ref{globalrel}) under the vanishing initial value assumption can be written as
\begin{subequations}\label{globalzeroini}
\be\label{globalc1j}
c_{1j}(t,k)=\Phi_{1j}(t,k)\ol{\phi_{2\times 2}(t,\bar k)}-\sig \Phi_{11}(t,k)\ol{\phi_{j1}(t,\bar k)}^{T}e^{2ikL},
\ee
\be\label{globalcj1}
c_{j1}(t,k)=\Phi_{j1}(t,k)\ol{\phi_{11}(t,\bar k)}-\sig \Phi_{2\times 2}(t,k)\ol{\phi_{1j}(t,\bar k)}^{T}e^{-2ikL}.
\ee

\end{subequations}
In view of the GLM representations (\ref{GLMFourier}), the global relation (\ref{globalc1j}) can be written as
\be\label{L1jglo}
-\hat{\tilde L}_{1j}+\sig \widehat{\bar{\tilde{\mathcal{L}}}}^{T}_{j1}e^{2ikL}=k\hat M_{1j}-k\sig \widehat{\bar{\mathcal{M}}}^{T}_{j1}e^{2ikL}+F_{1j}(k)-c_{1j}(k),
\ee
where
\be\label{F1jexpress}
F_{1j}=-\frac{i}{2}g_0 \hat M_{2\times 2}+\frac{i}{2}f_0\widehat{\bar{\mathcal{M}}}_{11}e^{2ikL}+\Phi_{1j}(\bar \phi_{2\times 2}^T-\id)-\sig(\Phi_{11}-1)\bar \phi_{j1}^Te^{2ikL}.
\ee
The expression of $F_{1j}$ can be expressed as in (\ref{F1jdef}). Letting $k\rightarrow -k$ in (\ref{L1jglo}), we find
\be\label{L1jglo2}
-\hat{\tilde L}_{1j}+\sig \widehat{\bar{\tilde{\mathcal{L}}}}^{T}_{j1}e^{-2ikL}=-k\hat M_{1j}+k\sig \widehat{\bar{\mathcal{M}}}^{T}_{j1}e^{-2ikL}+F_{1j}(-k)-c_{1j}(-k).
\ee
Solving (\ref{L1jglo}) and (\ref{L1jglo2}) for $\hat{\tilde L}_{1j}$, we find
\be
-\hat{\tilde L}_{1j}=\frac{2k\sig}{\Dta}\widehat{\bar{\mathcal{M}}}^{T}_{j1}-\frac{k\Sig}{\Dta}\hat M_{1j}-\frac{1}{\Dta}\left(e^{-2ikL}(F_{1j}-c_{1j})\right)_-
\ee

Multiplying this equation by $ke^{4ik^2(t-t')}$, $0<t'<t$, and integrating along $\pt D^{0}_1$ with respect to $dk$, we obtain
\be\label{L1jexp}
\ba{rl}
-\int_{\pt D^{0}_1}ke^{4ik^2(t-t')}\hat{\tilde L}_{1j}dk=&\int_{\pt D^{0}_1}\frac{2k^2\sig}{\Dta}e^{4ik^2(t-t')}\widehat{\bar{\mathcal{M}}}^{T}_{j1}dk-\int_{\pt D^{0}_1}\frac{k^2\Sig}{\Dta}e^{4ik^2(t-t')}\hat M_{1j}dk\\
&{}-\int_{\pt D^{0}_1}\frac{k}{\Dta}\left(e^{-2ikL}F_{1j}\right)_-dk
\ea
\ee
where we have used that the function $\frac{k}{\Dta}(e^{-2ikL}c_{1j})_-$ is bounded and analytic in $D^{0}_1$ so that its contributions vanish by Jordan's lemma.

The next is to take limit $t'\uparrow t$ in (\ref{L1jexp}). This can be achieved by using the identities
\be\label{kiden}
\int_{\pt D_1}ke^{4ik^2(t-t')}\hat f(t,k)dk=\left\{
\ba{ll}
\frac{\pi}{2}f(t,2t'-t),&0<t'<t,\\
\frac{\pi}{4}f(t,t),&0<t'=t,
\ea
\right.
\ee
and
\be\label{k2iden}
\ba{l}
\int_{\pt D^{0}_1}k^2e^{4ik^2(t-t')}\hat f(t,k)dk=\\
{}2\int_{\pt D^{0}_1}k^2\left[\int_{0}^{t'}e^{4ik^2(\tau-t')}f(t,2\tau-t)d\tau-\frac{f(t,2t'-t)}{4ik^2}\right]dk,\quad 0<t'<t.
\ea
\ee
The identity (\ref{k2iden}) is also valid if $k^2$ is replaced by $\frac{k^2}{\Dta}$ or $\frac{k^2\Sig}{\Dta}$ or $k^3$. Utilizing these identities
in (\ref{L1jexp}), we find
\be
\ba{rl}
-\frac{\pi}{2}\tilde L_{1j}(t,2t'-t)=&4\int_{\pt D^{0}_1}\frac{k^2\sig}{\Dta}\left[\int_{0}^{t'}\ol{\mathcal{M}}^{T}_{j1}(t,2\tau-t)e^{4ik^2(\tau-t')}d\tau-\frac{\ol{\mathcal{M}}^{T}_{j1}(t,2t'-t)}{4ik^2}\right]dk\\
&-2\int_{\pt D^{0}_1}\frac{k^2\Sig}{\Dta}\left[\int_{0}^{t'}M_{1j}(t,2\tau-t)e^{4ik^2(\tau-t')}d\tau-\frac{M_{1j}(t,2t'-t)}{4ik^2}\right]dk\\
&-\int_{\pt D^{0}_1}\frac{k}{\Dta}e^{4ik^2(t-t')}(e^{-2ikL}F_{1j})_-dk
\ea
\ee
Letting $t'\uparrow t$ in this equation and using the initial conditions (\ref{Goursatintdata}) as well as the following lemma, we find the representation in (\ref{g1rep}).

To prove the equation (\ref{f1rep}), we use the global relation (\ref{globalcj1}). Noticing that it can be written as
\be \label{mathcalF1j}
\widehat{\tilde{\mathcal{L}}}_{1j}-\sig \hat{\bar{\tilde{L}}}^{T}_{j1}e^{-2ikL}=-k\hat{\mathcal{M}}_{1j}+k\sig \hat{\bar{M}}^{T}_{j1}e^{-2ikL}+\mathcal{F}_{1j}-\sig \bar{c}^{T}_{j1}e^{-2ikL}.
\ee
Letting $k\rightarrow -k$ in (\ref{mathcalF1j}) we can get a new equation and solving these two equations for $\widehat{\mathcal{L}}_{1j}$.
\be
\widehat{\tilde{\mathcal{L}}}_{1j}=\frac{2k\sig}{\Dta}\hat{\bar{M}}^{T}_{j1}-\frac{k\Sig}{\Dta}\hat{\mathcal{M}}_{1j}+\frac{1}{\Dta}\left(e^{2ikL}(\mathcal{F}_{1j}-\sig \bar{c}^{T}_{j1}e^{-2ikL})\right)_-
\ee
Similar to the process of the proof of (\ref{g1rep}), we find that we also need the following lemma to get the representation (\ref{f1rep}).

\begin{lemma}\label{DtoNlemma}
\begin{subequations}
\be\label{F1jlemma}
\ba{rl}
\lim_{t'\uparrow t}\int_{\pt D^{0}_1}\frac{k}{\Dta}e^{4ik^2(t-t')}(e^{-2ikL}F_{1j})_-dk=&\int_{\pt D^{0}_1}\frac{k}{\Dta}(e^{-2ikL}F_{1j})_-dk\\
+\frac{ig_0}{2}\int_{\pt D^{0}_1}k\hat M_{2\times 2}dk&+\frac{g_0}{2i}\int_{\pt D^{0}_1}k\ol{\hat{\mathcal{M}}}^{T}_{2\times 2}dk
\ea
\ee
\be\label{mathcalF1jlemma}
\ba{rl}
\lim_{t'\uparrow t}\int_{\pt D^{0}_1}\frac{k}{\Dta}e^{4ik^2(t-t')}(e^{2ikL}\mathcal{F}_{1j})_-dk=&\int_{\pt D^{0}_1}\frac{k}{\Dta}(e^{2ikL}\mathcal{F}_{1j})_-dk\\
+\frac{if_0}{2}\int_{\pt D^{0}_1}k\hat{\mathcal{M}}_{2\times 2}dk&+\frac{f_0}{2i}\int_{\pt D^{0}_1}k\ol{\hat{M}}^{T}_{2\times 2}dk
\ea
\ee
\end{subequations}
\end{lemma}

\begin{proof}
The proof is similar to \cite{fl3}. And we prove (\ref{F1jlemma}), the proof of (\ref{mathcalF1jlemma}) is similar.
We write
\be\label{F1jlimit}
\ba{l}
\int_{\pt D^{0}_1}\frac{k}{\Dta}e^{4ik^2(t-t')}(e^{-2ikL}F_{1j})_-dk=\int_{\pt D^{0}_1}ke^{4ik^2(t-t')}\frac{ig_0}{2}\hat M_{2\times 2}dk\\
{}-\int_{\pt D^{0}_1}ke^{4ik^2(t-t')}\left(\hat{\tilde L}_{1j}-\frac{ig_0}{2}\hat M_{2\times 2}\right)\left(\ol{\hat{\tilde L}}^{T}_{2\times 2}-\frac{i\sig}{2}\ol{\hat{\mathcal{M}}}^{T}_{1j}f_0\right)dk\\
{}+\int_{\pt D^{0}_1}\frac{k^2\Sig}{\Dta}e^{4ik^2(t-t')}\left[\hat M_{1j}\left(\ol{\hat{\tilde{\mathcal{L}}}}^{T}_{2\times 2}-\frac{i\sig}{2}\ol{\hat{\mathcal{M}}}^{T}_{1j}f_0\right)+\left(\hat{\tilde L}_{1j}-\frac{ig_0}{2}\hat M_{2\times 2}\right)\ol{\hat{\mathcal{M}}}^{T}_{2\times 2}\right]dk\\
{}-\int_{\pt D^{0}_1}\frac{2k^2\sig}{\Dta}e^{4ik^2(t-t')}\left[\left(\ol{\hat{\bar{\tilde{L}}}}_{11}-\frac{ig_0}{2}\ol{\hat{\bar{M}}}_{j1}\right)\widehat{\bar{\mathcal{M}}}^{T}_{j1}+
\ol{\hat{\bar{M}}}_{11}\left(\widehat{\ol{\tilde{\mathcal{L}}}}^{T}_{j1}-\frac{i\sig f_0}{2}\widehat{\ol{\mathcal{M}}}_{11}\right)\right]dk\\
{}-\int_{\pt D^{0}_1}k^3e^{4ik^2(t-t')}\hat M_{1j}\ol{\hat{\mathcal{M}}}^{T}_{2\times 2}dk
\ea
\ee
The first integral on the right hand side of (\ref{F1jlimit}) yields the following contribution in the limit $t'\uparrow t$:
\[
\lim_{t'\uparrow t}\frac{ig_0}{2}\int_{\pt D^{0}_1}ke^{4ik^2(t-t')}\hat M_{2\times 2}dk=\lim_{t'\uparrow t}\frac{ig_0}{2}\frac{\pi}{2}M_{2\times 2}(t,2t'-t)=\frac{i\pi g_0(t)}{4}M_{2\times 2}(t,t).
\]
On the other hand, utilizing the second row of (\ref{kiden}),
\[
\frac{ig_0}{2}\int_{\pt D^{0}_1}k\hat M_{2\times 2}dk=\frac{i\pi g_0(t)}{8}M_{2\times 2}(t,t).
\]
Therefore,
\be\label{F1jlimi1term}
\lim_{t'\uparrow t}\frac{ig_0}{2}\int_{\pt D^{0}_1}ke^{4ik^2(t-t')}\hat M_{2\times 2}dk=\frac{ig_0}{2}\int_{\pt D^{0}_1}k\hat M_{2\times 2}dk+\frac{ig_0}{2}\int_{\pt D^{0}_1}k\hat M_{2\times 2}dk
\ee
The first term on the right hand side of (\ref{F1jlimi1term}) is the contribution obtained by taking the limit inside the integral; this term is included in the first term on the right hand side of (\ref{F1jlimit}). In addition to this term, there is also an additional term arising from the interchange of
the limit and the integration; this is the second term on the right hand side of (\ref{F1jlimit}).

We now consider the last integral on the right hand side of (\ref{F1jlimit}), which can be written as
\be\label{F1jlimitlast}
-\int_{\pt D^{0}_1}k^3e^{4ik^2(t-t')}\hat M_{1j}\ol{\hat{\mathcal{M}}}^{T}_{2\times 2}dk=-2\int_{\pt D^{0}_1}k^3\left(\int_{0}^{t}e^{4ik^2(\tau-t')}M_{1j}(t,2\tau-t)d\tau\right)\ol{\hat{\mathcal{M}}}^{T}_{2\times 2}dk
\ee
The right hand side of (\ref{F1jlimitlast}) equals
\be
-2\int_{\pt D^{0}_1}k^3\left(\int_{0}^{t'}e^{4ik^2(\tau-t')}M_{1j}(t,2\tau-t)d\tau-\frac{M_{1j}(t,2t'-t)}{4ik^2}\right)\ol{\hat{\mathcal{M}}}^{T}_{2\times 2}dk
\ee
Then taking the limit $t'\uparrow t$ in (\ref{F1jlimitlast}) and noticing the initial conditions (\ref{Goursatintdata}), we find
\be
-\lim_{t'\uparrow t}\int_{\pt D^{0}_1}k^3e^{4ik^2(t-t')}\hat M_{1j}\ol{\hat{\mathcal{M}}}^{T}_{2\times 2}dk=-\int_{\pt D^{0}_1}k^3\hat M_{1j}\ol{\hat{\mathcal{M}}}^{T}_{2\times 2}dk+\frac{g_0(t)}{2i}\int_{\pt D^{0}_1}k\ol{\hat{\mathcal{M}}}^{T}_{2\times 2}dk
\ee
The first term on the right hand side is the contribution obtained by taking the limit inside the
integral. In addition to this term, there is also an additional term arising from the
interchange of the limit and the integration; this is the third term on the right hand side of
(\ref{F1jlimit}).

Finally, we claim that the limits of the second, third, and fourth integrals on the
right hand side of (\ref{F1jlimit}) can be computed by simply taking the limit inside the integral, i.e. in
these cases no additional terms arise. In fact, the second integral is the direct result by taking the limit inside the integral. We show the claim for the term
\[
I=\int_{\pt D^{0}_1}\frac{k^2\Sig}{\Dta}e^{4ik^2(t-t')}\hat M_{1j}\left(\ol{\hat{\tilde{\mathcal{L}}}}^{T}_{2\times 2}-\frac{i\sig}{2}\ol{\hat{\mathcal{M}}}^{T}_{1j}f_0\right)dk
\]
the proofs for the other terms are similar. We have
\[
I=2\int_{\pt D^{0}_1}\frac{k^2\Sig}{\Dta}\left(\int_{0}^{t}e^{4ik^2(\tau-t')}M_{1j}(t,2\tau-t)d\tau\right)\left(\ol{\hat{\tilde{\mathcal{L}}}}^{T}_{2\times 2}-\frac{i\sig}{2}\ol{\hat{\mathcal{M}}}^{T}_{1j}f_0\right)dk
\]
This can be written as
\[
\footnotesize
I=2\int_{\pt D^{0}_1}\frac{k^2\Sig}{\Dta}\left(\int_{0}^{t'}e^{4ik^2(\tau-t')}M_{1j}(t,2\tau-t)d\tau-\frac{M_{1j}(t,2t'-t)}{4ik^2}\right)\left(\ol{\hat{\tilde{\mathcal{L}}}}^{T}_{2\times 2}-\frac{i\sig}{2}\ol{\hat{\mathcal{M}}}^{T}_{1j}f_0\right)dk
\]
\normalsize
Taking the limit $t'\uparrow t$, we find
\[
\ba{l}
\lim_{t'\uparrow t}\int_{\pt D^{0}_1}\frac{k^2\Sig}{\Dta}e^{4ik^2(t-t')}\hat M_{1j}\left(\ol{\hat{\tilde{\mathcal{L}}}}^{T}_{2\times 2}-\frac{i\sig}{2}\ol{\hat{\mathcal{M}}}^{T}_{1j}f_0\right)dk=\\
\int_{\pt D^{0}_1}\frac{k^2\Sig}{\Dta}\left(\hat{M}_{1j}(t,k)-\frac{g_0(t)}{2ik^2}\right)\left(\ol{\hat{\tilde{\mathcal{L}}}}^{T}_{2\times 2}-\frac{i\sig}{2}\ol{\hat{\mathcal{M}}}^{T}_{1j}f_0\right)dk
\ea
\]
However, in this case the additional term
\[
-\int_{\pt D^{0}_1}\frac{k^2\Sig}{\Dta}\frac{g_0(t)}{2ik^2}\left(\ol{\hat{\tilde{\mathcal{L}}}}^{T}_{2\times 2}-\frac{i\sig}{2}\ol{\hat{\mathcal{M}}}^{T}_{1j}f_0\right)dk
\]
vanishes because the integrand is analytic and goes to zero as $k\rightarrow \infty$ in $D_1$.

\end{proof}

{\bf Return to prove theorem \ref{maintheorem}.} We now consider the Neumann problem. Solving (\ref{L1jglo}) and (\ref{L1jglo2}) for $\hat M_{1j}$, we find
\be
k\hat M_{1j}=\frac{\Sig}{\Dta}\hat{\tilde{L}}_{1j}-\frac{2\sig}{\Dta}\widehat{\bar{\tilde{\mathcal{L}}}}^{T}_{j1}+\frac{1}{\Dta}\left(e^{-2ikL}(F_{1j}-c_{1j})\right)_+.
\ee
Similar to solve (\ref{mathcalF1j}) and related equation for $\hat{\mathcal{M}}_{1j}$, we find
\be
k\widehat{\mathcal{M}}_{1j}=-\frac{\Sig}{\Dta}\hat{\tilde{\mathcal{L}}}_{1j}+\frac{2\sig}{\Dta}\widehat{\bar{\tilde{{L}}}}^{T}_{j1}+\frac{1}{\Dta}\left(e^{2ikL}(\mathcal{F}_{1j}-\sig \bar{c}^{T}_{j1}e^{-2ikL})\right)_+.
\ee
Multiplying this equation by $e^{4ik^2(t-t')}$, $0<t'<t$, and integrating along $\pt D^{0}_1$ with respect to $dk$, we obtain
\begin{subequations}
\be
\frac{\pi}{2}M_{1j}(t,2t'-t)=\int_{\pt D^0_1}e^{4ik^2(t-t')}\left\{\frac{\Sig}{\Dta}\hat{\tilde{L}}_{1j}-\frac{2\sig}{\Dta}\widehat{\bar{\tilde{\mathcal{L}}}}^{T}_{j1}+\frac{1}{\Dta}\left(e^{-2ikL}F_{1j}\right)_+\right\}dk
\ee
\be
\frac{\pi}{2}\mathcal{M}_{1j}(t,2t'-t)=\int_{\pt D^0_1}e^{4ik^2(t-t')}\left\{-\frac{\Sig}{\Dta}\hat{\tilde{\mathcal{L}}}_{1j}+\frac{2\sig}{\Dta}\widehat{\bar{\tilde{{L}}}}^{T}_{j1}+\frac{1}{\Dta}\left(e^{2ikL}\mathcal{F}_{1j}\right)_+\right\}dk
\ee
\end{subequations}
where we used that the functions
\[
\frac{1}{\Dta}\left(e^{-2ikL}c_{1j}\right)_+,\quad \frac{1}{\Dta}\left(\bar{c}^{T}_{j1}\right)_+
\]
are bounded and analytic in $D^{0}_1$ so that its contributions vanish by Jordan's lemma. Letting $t'\uparrow t$ in this equation and using the initial conditions (\ref{Goursatintdata}) as well as the following lemma, we find the representation in (\ref{g0f0rep}).

\begin{lemma}
\begin{subequations}
\be
\lim_{t'\uparrow t}\int_{\pt D^{0}_1}\frac{1}{\Dta}e^{4ik^2(t-t')}\left(e^{-2ikL}F_{1j}\right)_+dk=\int_{\pt D^{0}_1}\frac{1}{\Dta}\left(e^{-2ikL}F_{1j}\right)_+dk
\ee
\be
\lim_{t'\uparrow t}\int_{\pt D^{0}_1}\frac{1}{\Dta}e^{4ik^2(t-t')}\left(e^{2ikL}\mathcal{F}_{1j}\right)_+=\int_{\pt D^{0}_1}\frac{1}{\Dta}\left(e^{2ikL}\mathcal{F}_{1j}\right)_+
\ee
\end{subequations}
\end{lemma}
\begin{proof}
The proof of this lemma is similar to the proof of Lemma \ref{DtoNlemma}
\end{proof}

Hence, we complete the proof of theorem \label{maintheorem}.

\end{proof}

\subsection{Equivalence of the two representations}
We will show that the representations derived using the GLM approach in theorem \ref{maintheorem} coincide with those of theorem 4.3 in \cite{jf3}.

\begin{itemize}
  \item {\bf The Dirichlet Problem, i.e., the representations for $g_1(t)$ and $f_1(t)$.} Using the expression (\ref{F1jexpress}) for $F_{1j}$ as well as the formulas
      \[
      \hat M_{ij}=\frac{1}{2k}\Phi_{ij,-},\quad \widehat{\mathcal{M}}_{ij}=\frac{1}{2k}\phi_{ij,-},\quad i,j=1,2,3,
      \]
      we can write the representation of (\ref{g1rep}) as
      \be
      \ba{rl}
      g_1(t)=&\frac{4}{i\pi}\int_{\pt D^{0}_1}\left\{\frac{\Sig}{2\Dta}\left[k\Phi_{1j,-}+ig_0(t)\right]-\frac{\sig}{\Dta}\left[k\bar{\phi}^{T}_{j1,-}+i\sig f_0(t)\right]\right.\\
      &+\frac{ig_0(t)}{4}\Phi_{2\times 2,-}+\frac{g_0(t)}{4i}\bar{\phi}^T_{2\times 2,-}-\frac{k}{\Dta}\left[\frac{ig_0(t)}{4k}\Phi_{2\times 2,-}e^{-2ikL}\right]_-\\
      &\left.-\frac{k\sig}{\Dta}\left[(\Phi_{11}-1)\bar{\phi}^T_{j1}\right]_-+\frac{k}{\Dta}\left[\Phi_{1j}(\bar{\phi}^{T}_{2\times 2}-\id)e^{-2ikL}\right]_-
      \right\}dk.
      \ea
      \ee
      In view of the identity
      \[
      -\frac{k}{\Dta}\left[\frac{ig_0(t)}{4k}\Phi_{2\times 2,-}e^{-2ikL}\right]_-+\frac{ig_0(t)}{4}\Phi_{2\times 2,-}=\frac{ig_0(t)}{2}\Phi_{2\times 2,-},
      \]
      and recalling the definition of the $3\times 3$ partition matrix, we can obtain the same results (4.30) in \cite{jf3}.  Similar computations show that the representations for $f_1$ are also equivalent.

      \item {\bf The Neumann Problem, i.e., the representations for $g_0(t)$ and $f_0(t)$.} Using the expression (\ref{F1jexpress}) for $F_{1j}$ as well as the formulas
      \[
      \Phi_{1j,+}=2\hat{\tilde{L}}_{1j}-ig_0(t)\hat M_{2\times 2},\quad \phi_{1j,+}=2\widehat{\tilde{\mathcal{L}}}_{1j}-i\sig f_0\widehat{\mathcal{M}}_{2\times 2},
      \]
      a straightforward computation shows that the representations of (\ref{g0f0rep}) and the equations (4.31) in \cite{jf3} are equivalent.
\end{itemize}

\bigskip

{\bf Acknowledgements}
This work of Xu was supported by National Natural Science Foundation of China under project NO.11501365, Shanghai Sailing Program
supported by Science and Technology Commission of Shanghai Municipality
under Grant NO.15YF1408100, Shanghai youth teacher assistance program NO.ZZslg15056 and the Hujiang Foundation of China (B14005). Fan was support by grants from the National
Science Foundation of China (Project No.10971031; 11271079; 11075055).


\begin{thebibliography}{XXXX}

\bibitem{m74} S.V. Manakov, On the theory of two-dimensional stationary self-focusing of electromagenic waves, Sov. Phys. JETP, 38(1974) 248-253.

\bibitem{ba2001} T.H. Busch, J.R. Anglin, Dark-bright solitons in inhomogeneous Bose-Einstein
condensates, Phys. Rev. Lett. 87(2001) 010401.

%
%

\bibitem{f1} A. S. Fokas, A unified transform method for solving linear and certain nonlinear PDEs, Proc. R. Soc. Lond. A 453(1997), 1411-1443.

\bibitem{f2} A. S. Fokas, On the integrability of linear and nonlinear partial differential equations, J. Math. Phys.
41(2000) 4188-4237.

\bibitem{f3} A. S. Fokas, Integrable nonlinear evolution equations on the half-line, Commun. Math. Phys. 230(2002), 1-39.

\bibitem{f4} A.S. Fokas, A Unified Approach to Boundary Value Problems, in: CBMS-NSF Regional Conference Series in Applied Mathematics, SIAM, 2008.

\bibitem{fi1} A.S. Fokas, A.R. Its, An initial-boundary value problem for the Korteweg-de Vries equation. Math.
Comput. Simul. 37(1994) 293-321.

\bibitem{fi2} A.S. Fokas, A.R. Its, An initial-boundary value problem for the sine-Gordon equation. Theor.Math.
Physics 92(1992) 388-403.

\bibitem{fi3} A.S. Fokas, A.R. Its, The linearization of the initial-boundary value problem of the nonlinear
Schr¡§odinger equation. SIAM J. Math. Anal. 27(1996) 738-764.

\bibitem{fis} A.S. Fokas, A.R. Its, L.Y. Sung, The nonlinear Schr\"odinger equation on the half-line. Nonlinearity
18(2005) 1771-1822.

\bibitem{fi4} A.S. Fokas, A.R. Its, The nonlinear Schr¡§odinger equation on the interval. J. Phys.A 37(2004) 6091-6114.

\bibitem{abmfs1} A. Boutet De Monvel, A.S. Fokas, D. Shepelsky, Integrable nonlinear evolution equations on a finite interval, Comm. Math. Phys. 263(2006) 133-172.

\bibitem{abmfs2} A. Boutet de Monvel,A.S.Fokas,D.Shepelsky, The mKDV equation on the half-line, J. Inst. Math. Jussieu. 3(2004), 139-164.
%
%
%
\bibitem{fcpam} A.S. Fokas, The Generalized Dirichlet-to-Neumann Map for Certain Nonlinear Evolution PDEs, Commun. Pure Appl. Math., 58(2005) 639-670.
%

\bibitem{l1} J. Lenells, Initial-boundary value problems for integrable evolution equations with $3\times 3$ Lax pairs, Physica D 241(2012) 857-875.

%
%

\bibitem{jf3} J. Xu, E. Fan, Initial-boundary value problem for integrable nonlinear evolution equations with $3\times 3$ Lax pairs on the interval, to apperear in Stud. Appl. Math.





\bibitem{fl3} J. Lenells, A. S. Forkas, The unified method: \Rmnum{3}. Nonlinearizable problem on the interval, J. Phys. A: Math. Theor. 45(2012) 195203;











\end{thebibliography}
\end{document}